\documentclass{article}

\usepackage{arxiv}

\usepackage[utf8]{inputenc} 
\usepackage[T1]{fontenc}    
\usepackage{hyperref}       
\usepackage{url}            
\usepackage{booktabs}       
\usepackage{amsfonts}       
\usepackage{nicefrac}       
\usepackage{microtype}      
\usepackage{amsmath}
\usepackage{amsthm}
\usepackage{amssymb}
\usepackage{dsfont }
\usepackage{bbm}
\usepackage{cleveref}       
\usepackage{lipsum}         
\usepackage{graphicx}
\usepackage{natbib}
\usepackage{doi}
\usepackage{enumitem}  
\usepackage{threeparttable}

\newtheorem{theorem}{Theorem}
\newtheorem{lemma}{Lemma}

\newtheorem{assumption}{Assumption}

\newcommand{\E}{\mathbb{E}}
\newcommand{\Var}{\mathrm{Var}}
\newcommand{\Cov}{\mathrm{Cov}}
\newcommand{\Prob}{\mathbb{P}}
\newcommand{\cudml}{\widehat{\theta}^{\text{CU-DML}} }
\newcommand{\norm}[1]{\left\lVert #1 \right\rVert}

\title{Calibrating doubly-robust estimators with unbalanced treatment assignment}

\date{May 20, 2024}

\author{ Daniele Ballinari\thanks{The views, opinions, findings, and conclusions or recommendations expressed in this paper are strictly those of the author(s). They do not necessarily reflect the views of the Swiss National Bank (SNB). The SNB takes no responsibility for any errors or omissions in, or for the correctness of, the information contained in this paper.}\thanks{For helpful comments and suggestions, we thank Nora Bearth, Michael Lechner, Thomas Maag, Alexander Wehrli, Michael Zimmert, and two anonymous reviewers.}\\
    Swiss National Bank\\
    Börsenstrasse 15\\
    8001 Zurich, Switzerland\\
    \texttt{daniele.ballinari@snb.ch} \\
}

\renewcommand{\headeright}{Working paper}
\renewcommand{\undertitle}{Working paper}
\renewcommand{\shorttitle}{Calibrating doubly-robust estimators with unbalanced treatment assignment}

\begin{document}
\maketitle

\begin{abstract}
Machine learning methods, particularly the double machine learning (DML) estimator \citep{Chernozhukov2018}, are increasingly popular for the estimation of the average treatment effect (ATE). However, datasets often exhibit unbalanced treatment assignments where only a few observations are treated, leading to unstable propensity score estimations. We propose a simple extension of the DML estimator which undersamples data for propensity score modeling and calibrates scores to match the original distribution. The paper provides theoretical results showing that the estimator retains the DML estimator's asymptotic properties. A simulation study illustrates the finite sample performance of the estimator.
\end{abstract}

\jelcodes{C14 \and C21 \and C52 \and C55}

\section{Introduction}
Estimation of the average treatment effect (ATE) is of central importance in empirical research. The interest generally lies in the effect that a binary treatment has on an outcome variable. For example, we might be interested in the effect that a training program has on the unemployment duration. With the increasing availability of large datasets, the use of machine learning (ML) methods to estimate the ATE has become popular \citep{Athey2019,Liuyi2021}. The double machine learning (DML) estimator \citep{Chernozhukov2018} is a widely adopted ATE estimator that relies on ML-estimated nuisance functions. The DML estimator has been shown to be consistent and asymptotically normal, even when using ML methods that converge at a slower rate than the parametric rate \citep{Chernozhukov2018}.

While dataset sizes have increased, the number of treated units often remains small.
The treatment is often costly and/or time-consuming. For example, the submission of a new drug to a patient might take several months, or a training program for the unemployed is very costly. Control outcomes and covariates, on the other hand, are more easily collected from, e.g., administrative agencies, medical records, or financial markets \citep{Kunzel2019,Bouchaud2022,Hujer2006}. The dataset might then consist of only a few treated, but many control observations. This unbalancedness can lead to unstable propensity score estimations \citep{Huber2013}.
Even though the DML estimator relies on a doubly robust approach that combines the conditional outcome expectations with the propensity score, instability in the propensity score estimation can lead to high variability in the ATE estimate.

In this paper, we propose a simple extension of the DML estimator that addresses the issue of unbalanced treatment assignment. Inspired by the ML classification literature \citep{Japkowicz2002}, the proposed approach undersamples the data used for fitting the ML model for the propensity score. Using the relation between the true and the undersampled propensity score, we calibrate the propensity scores to match the original data distribution. We show that the proposed estimator has the same asymptotic distribution as the DML estimator, attains the parametric rate of convergence $\sqrt{N}$ and its variance achieves the semi-parametric efficiency bound \citep{Hahn1998}. We illustrate the finite sample performance of the estimator in a simulation study. While we present the results in the context of the DML estimator, the proposed approach applies to any ATE estimator that relies on the efficient score function.

\section{Calibration estimator}
\subsection{Notation and causal identification}
We define causal effects using Rubin's \citeyear{Rubin1972} potential outcome framework. The interest lies in the effect of a binary treatment variable\footnote{The estimator presented in this paper can be directly generalized to multivalued treatments \citep[see, among others,][]{Farrell2015,Knaus2020}.} $D$ on an outcome variable $Y$. We denote the potential outcome for treatment $D=d$, that is the outcome one would observe if the treatment was $d$, as $Y^d$. The effect of interest is the average treatment effect (ATE) defined as:
\begin{equation}\label{eq:ate}
    \theta = \E[Y^1 - Y^0]
\end{equation}

The potential outcomes are not directly observed and the parameter of interest has to be identified from observational data. The researcher observes a sample of i.i.d. random variables $\{Z_1, \dots, Z_N\}$ where $Z_i := \left(X_i, D_i, Y_i\right)$, with $X_i$ being a $p$-dimensional vector of exogenous control variables with support $\mathcal{X}$, $D_i$ the binary treatment random variable, and $Y_i$ the observed outcome.  Define the conditional outcome expectations as $\mu_d(X) = \E[Y\vert D=d, X]$ and the propensity score as $p(X) = \E[D|X] = \Prob[D=1|X]$.
The respective estimated quantities are denoted as $\widehat{\mu}_d(X)$ and $\widehat{p}(X)$. Finally, we denote by $\tau(Z)$ the efficient score function:
\begin{equation*}
    \tau(Z) =  \mu_1(X) - \mu_0(X) + \frac{D}{p(X)}(Y-\mu_1(X)) - \frac{1-D}{1-p(X)}(Y-\mu_0(X)).
\end{equation*}

Identification of the ATE from observable outcomes is achieved under the following assumption.
\begin{assumption}[Identification]\label{assumption:identification} For observation $Z = \left(X, D, Y\right)$ assume that \citep{Rosenbaum1983}:
    \begin{enumerate}[label=(\roman*)]
        \item $Y^0, Y^1 \perp D \vert X = x $ for any $x \in \mathcal{X}$ (conditional independence assumption).
        \item The observed outcome is $Y = D Y^1 + (1-D) Y^0$ (stable unit treatment value assumption).
        \item For any $x \in \mathcal{X}$ it holds that $\eta < p(x) < 1-\eta$ for some $\eta>0$ (common support).
        \item $X = X^1 = X^0$ where $X^d$ denotes the random covariate vector under treatment $d$ (exogenity of covariates).
    \end{enumerate}
\end{assumption}
Under Assumption \ref{assumption:identification} the ATE can be characterized as a functional of the joint distribution of the observed data $(X,D,Y)$ \citep{Athey2019}:
\begin{align}\label{eq:dr}
    \theta = \E[Y^1 - Y^0] =  \E[\mu_1(X) - \mu_0(X)] = \E\left[ \tau(Z) \right].
\end{align}
 
\subsection{Double machine learning estimator}
The estimator proposed in this study extends the popular ATE estimator developed by \cite{Chernozhukov2018} generally referred to as the double machine learning (DML) estimator. DML builds on two key ingredients. First, it uses the efficient score function $\tau(Z)$ to construct an estimator of the ATE. Second, it uses cross-fitting to estimate the nuisance functions $\mu_d(X)$ and $p(X)$. In more detail, the DML estimator for the ATE is defined as follows:

\begin{enumerate}[label*=\textsc{Step} \arabic*]
    \item For some fixed $K \in \{2, \dots, N\}$, randomly partition the observation indices into $K$ sets $\mathcal{I}_1, \dots, \mathcal{I}_K$ of equal size. Denote the complement of $\mathcal{I}_k$ by $\mathcal{I}_{-k} = \{1, \dots, N\}\setminus\mathcal{I}_k$. Denote the cardinality of each set of indices by $\vert \mathcal{I}_k \vert$.
    \item \textbf{for} $k=1$ \textbf{to} $K$ \textbf{do}:\\[0.2cm]
    \indent Estimate the nuisance functions $\mu_d(x)$ and $p(x)$ on the sample defined by indices $\mathcal{I}_{-k}$ and denote the estimated functions by $\widehat{\mu}_d^{\mathcal{I}_{-k}}(x)$ and $\widehat{p}^{\mathcal{I}_{-k}}(x)$.\\[0.2cm]
    \textbf{end for}
    \item Estimate the ATE using the estimator:
    \begin{equation}
        \widehat{\theta}^{\text{DML}} = \frac{1}{N}\sum_{k=1}^{K} \sum_{i\in\mathcal{I}_k} \widehat\tau^{\mathcal{I}_{-k}}(Z_i)
    \end{equation}
    where:
    \begin{equation*}
        \widehat\tau^{\mathcal{I}_{-k}}(Z_i) = \widehat\mu^{\mathcal{I}_{-k}}_1(X_i) - \widehat\mu^{\mathcal{I}_{-k}}_0(X_i) + \frac{D_i}{\widehat{p}^{\mathcal{I}_{-k}}(X_i)}(Y_i-\widehat\mu^{\mathcal{I}_{-k}}_1(X_i)) - \frac{1-D_i}{1-\widehat{p}^{\mathcal{I}_{-k}}(X_i)}(Y_i-\widehat\mu^{\mathcal{I}_{-k}}_0(X_i))
    \end{equation*}
\end{enumerate}

\cite{Chernozhukov2018} provide an asymptotic theory for the DML estimator. In particular, they show that the estimator is consistent and asymptotically normal, even when using machine learning (ML) methods that converge at a slower rate than the parametric rate.

\subsection{Calibration estimator}
A shortcoming of the DML estimator is its poor finite sample performance when the treatment assignment is unbalanced, i.e. when either very few or very many observations are treated. In the following, without loss of generality, we consider only the case where very few are treated. ML models perform poorly when the data is unbalanced \citep{Japkowicz2002} and predict propensity scores that are potentially close to zero or one. A common approach to improving the performance of ML models is to undersample the observations that are overrepresented, i.e. the observations that are not treated \citep{He2009}.\footnote{Another common technique used in the ML literature is to oversample the minority class \citep{Japkowicz2002}. This approach might be problematic in the context of DML since it introduces dependence in the data, complicating the asymptotic theory of the proposed estimator. The convergence rates of common ML methods needed for the asymptotic results presented in Section \ref{subsec:asymptotic_results} have been proven for the case of independent and identically distributed data; see, e.g. \cite{Belloni2013} for the Lasso, \cite{Luo2016} for $L_2$ boosting, \cite{Wager2016} for random forests, and \cite{Chen1999} for neural networks.}

Undersampling the observations changes the underlying distribution of the data and the predicted propensity scores will be biased for the true propensity scores from the unbalanced distribution \citep{Pozzolo2015}. To formalize this concept, let $S_i$ be a random variable equal to 1 if observation $i$ is sampled and 0 otherwise. It then follows that $\Prob[S_i=1|D_i=1]=1$ since all treated observations are kept in the sample. For the untreated observations we have that $\Prob[S_i=1|D_i=0] = \E[D_i]/(1-\E[D_i]) =: \gamma$. Moreover, since the undersampling strategy is not dependent on the control variables $X$, we have that $\Prob[S_i=1\vert D_i=d, X_i] = \Prob[S_i=1\vert D_i=d]$. Using Bayes' rule it can be shown that the propensity score for the undersampled data $p_S(X_i)$ is given by \citep{Pozzolo2015}:
\begin{equation}\label{eq:undersampled_propensity}
    p_S(X_i) = \frac{p(X_i)}{p(X_i) + \gamma \cdot \left(1-p(X_i)\right)}
\end{equation}
from where it follows that $p_S(X_i)\neq p(X_i)$ for $\gamma < 1$.

A na\"ive solution to address this issue would be to not only undersample the data used for estimating the propensity score but also the data on which the efficient score function is computed. In other words, a valid strategy would be to undersample the entire dataset and apply DML to the undersampled data (hereafter referred to as U-DML). However, this approach reduces the number of observations from $N$ to $2\cdot \E[D_i] \cdot N$. In cases where only 5\% of the observations are treated, we would discard 90\% of the observations. 

We propose a calibration estimator that uses the entire sample and corrects for the bias in the propensity score. The main idea is to only undersample the data used for fitting the ML model for the propensity score. Using the relation between the true and the undersampled propensity score in Equation \eqref{eq:undersampled_propensity}, we calibrate the propensity scores to match the original data distribution. In more detail, the calibrated-undersampled DML (CU-DML) estimator is defined as follows:

\begin{enumerate}[label*=\textsc{Step} \arabic*]
    \item Estimate $\gamma$ as:
    \begin{equation}\label{eq:odd_estimator}
        \widehat{\gamma} = \frac{\sum_{i=1}^{N} D_i}{\sum_{i=1}^{N} (1-D_i)}
    \end{equation}
    \item For some fixed $K \in \{2, \dots, N\}$, randomly partition the observation indices into $K$ sets $\mathcal{I}_1, \dots, \mathcal{I}_K$ of equal size. Denote the complement of $\mathcal{I}_k$ by $\mathcal{I}_{-k} = \{1, \dots, N\}\setminus\mathcal{I}_k$. Denote the cardinality of each set of indices by $\vert \mathcal{I}_k \vert$.
    \item \textbf{for} $k=1$ \textbf{to} $K$ \textbf{do}:\\[0.2cm]
    \begin{enumerate}[label*=.\arabic*]
        \item Estimate the nuisance functions $\mu_d(X)$ on the sample defined by indices $\mathcal{I}_{-k}$ and denote the estimated functions by $\widehat{\mu}_d^{\mathcal{I}_{-k}}(X)$.
        \item Draw random variables $S_i$ for $i \in \mathcal{I}_{-k}$ from a Bernoulli distribution with probabilitiy $\Prob[S_i=1|D_i=1] = 1$ and $\Prob[S_i=1|D_i=0] = \widehat\gamma$. Define the undersampled indices as:
        \begin{equation*}
            \mathcal{I}_{-k}^S = \{ i \in \mathcal{I}_{-k} \vert S_i=1 \}
        \end{equation*}
        Estimate the nuisance function $p_S(X)$ on the sample defined by indices $\mathcal{I}^S_{-k}$ and denote the estimated function by $\widehat{p}_S^{\mathcal{I}^S_{-k}}(X)$.
        Calibrate the propensity score estimated on the undersampled data to match the original data distribution:
        \begin{equation}\label{eq:propensity_estimator}
            \widehat{p}^{\mathcal{I}_{-k}}(X) = \frac{\widehat{\gamma}\cdot \widehat{p}_S^{\mathcal{I}^S_{-k}}(X)}{\widehat{\gamma} \cdot \widehat{p}_S^{\mathcal{I}^S_{-k}}(X) + \left(1-\widehat{p}_S^{\mathcal{I}^S_{-k}}(X)\right)}
        \end{equation}
    \end{enumerate}
    \textbf{end for}
    \item Estimate the ATE using the estimator:
    \begin{equation}
        \widehat{\theta}^{\text{CU-DML}} = \frac{1}{N}\sum_{k=1}^{K} \sum_{i\in\mathcal{I}_k} \widehat\tau^{\mathcal{I}_{-k}}(Z_i)
    \end{equation}
    where:
    \begin{equation*}
        \widehat\tau^{\mathcal{I}_{-k}}(Z_i) = \widehat\mu^{\mathcal{I}_{-k}}_1(X_i) - \widehat\mu^{\mathcal{I}_{-k}}_0(X_i) + \frac{D_i}{\widehat{p}^{\mathcal{I}_{-k}}(X_i)}(Y_i-\widehat\mu^{\mathcal{I}_{-k}}_1(X_i)) - \frac{1-D_i}{1-\widehat{p}^{\mathcal{I}_{-k}}(X_i)}(Y_i-\widehat\mu^{\mathcal{I}_{-k}}_0(X_i))
    \end{equation*}
\end{enumerate}

Notice that, in contrast to the nuisance functions, $\gamma$ can be estimated from the entire sample. For $\widehat{\gamma} = 1$ the CU-DML estimator reduces to the classical DML estimator. Furthermore, the results presented in the next section show that asymptotically $\widehat{\theta}^{\text{CU-DML}}$ converges in probability to $\widehat{\theta}^{\text{DML}}$, also for $\widehat{\gamma} < 1$.

\subsection{Asymptotic results}\label{subsec:asymptotic_results}
We now present the asymptotic results for the CU-DML estimator. We start by introducing the following assumptions.

\begin{assumption}[Boundedness of conditional variances]\label{assumption:variance} For the conditional variance of the outcome it holds that:
    \begin{equation*}
        \sup_{x\in\mathcal{X}} \Var\left[Y|D=d,X=x\right] < \zeta
    \end{equation*}
    for some $\zeta < \infty$.
\end{assumption}

\begin{assumption}[Boundedness of the propensity score]\label{assumption:propensity_score}\phantom{empty}
    \begin{enumerate}[label=(\roman*)]
        \item\label{assumption:propensity_score_1} For the unconditional propensity score $\lambda:=\E[p(X)]=\E[D]$ and its estimator $\widehat{\lambda} = N^{-1} \sum_{i=1}^N D_i$ it holds that:
        \begin{equation*}
            \epsilon < \lambda < 1 - \epsilon \qquad \epsilon < \widehat{\lambda} < 1 - \epsilon
        \end{equation*}
        for some $\epsilon > 0$.
        \item\label{assumption:propensity_score_2} For all $x\in\mathcal{X}$ it holds that:
        \begin{equation*}
            \eta < p_S(x) < 1 - \eta
        \end{equation*}
        for some $\eta > 0$.
    \end{enumerate}
\end{assumption}

\begin{assumption}[Convergence of the ML estimators]\label[type]{assumption:convergence_ml}\phantom{empty}
    \begin{enumerate}[label=(\roman*)]
        \item\label{assumption:convergence_ml_1} The ML methods are sup-norm consistent:
        \begin{equation*}
            \sup_{x\in\mathcal{X}} \vert \widehat\mu_d(x) - \mu_d(x) \vert \overset{p}{\rightarrow} 0 \qquad \sup_{x\in\mathcal{X}} \vert \widehat{p}_S(x) - p_S(x) \vert \overset{p}{\rightarrow} 0
        \end{equation*}
        \item\label{assumption:convergence_ml_2} The ML methods have risk-decay rates that satisfy (risk-decay assumption):
        \begin{equation*}
            \E\left[\left(\widehat\mu_d(x) - \mu_d(x)\right)^2 \right] \E\left[\left(\widehat{p}_S(x) - p_S(x)\right)^2 \right] = o(N^{-1})
        \end{equation*}
    \end{enumerate}
\end{assumption}

These assumptions closely resemble the ones required for the asymptotic results of the DML estimator \citep{Chernozhukov2018,Wager2022}. In addition to the boundedness of the propensity score, Assumption~\ref{assumption:propensity_score}~\ref{assumption:propensity_score_1} requires the unconditional propensity score to be bounded away from zero and one. This assumption is natural, as in a situation where the expected propensity score equals zero (one), there are no (only) treated. Importantly, Assumption~\ref{assumption:propensity_score}~\ref{assumption:propensity_score_2} relaxes the usual bounds imposed on the propensity score, since $\tilde\eta<p(X)<1-\tilde\eta$ with $\tilde\eta<\eta$ whenever $\gamma<1$.\footnote{From Equation \eqref{eq:undersampled_propensity} it follows that $\tilde\eta<p(X)<1-\tilde\eta$ with $\tilde\eta = \gamma\eta/(1-\eta+\gamma\eta)$.} Assumption~\ref{assumption:convergence_ml} states the convergence rate requirements of the ML estimators in terms of the undersampled propensity score $p_S(X)$. The theoretical result is then given by the following theorem.

\begin{theorem}
    \label{theorem:convergence}
    Under Assumptions \ref{assumption:variance}, \ref{assumption:propensity_score} and \ref{assumption:convergence_ml} it holds that:
    \begin{equation*}
        \sqrt{N} \left( \widehat{\theta}^{\text{CU-DML}} - \theta \right) \overset{d}{\longrightarrow} \mathcal{N}\left(0, V^*\right)
    \end{equation*}
    where:
    \begin{equation*}
        V^* = \Var[\mu_1(X_i)-\mu_0(X_i)] + \E\left[\frac{\sigma_1^2(X_i)}{p(X_i)}\right] + \E\left[\frac{\sigma_0^2(X_i)}{1-p(X_i)}\right]
    \end{equation*}
    with $\sigma^2_d(X_i) = \Var[Y_i^d\vert X_i]$.
\end{theorem}

The proof of Theorem \ref{theorem:convergence} is relegated to Appendix \ref{app:proofs}. Theorem \ref{theorem:convergence} shows that the CU-DML estimator has the same asymptotic distribution as the DML estimator. In particular, the estimator attains the parametric rate of convergence $\sqrt{N}$ and its variance achieves the semi-parametric efficiency bound \citep{Hahn1998}.

\section{Simulation study}
In this section, we investigate the finite sample performance of the CU-DML estimator in a simulation study. We consider two simulation strategies. In the first strategy, we generate data from a synthetic data generating process (DGP). In more detail, we generate data from the DGP used by \cite{Nie2020}:
\begin{align}
    X_i &\sim \text{Unif}[0,1]^{20}, \quad D_i \vert X_i \sim \text{Bernoulli}\left(p(X_i)\right), \quad \epsilon_i \overset{\text{iid}}{\sim} \mathcal{N}(0,1), \notag \\
    Y_i &= b(X_i) + (D_i-0.5) \left(X_{i,1} + X_{i,2}\right) + \sigma \epsilon_i.\label{eq:synthetic_dgp}
\end{align}
The baseline main effect is the scaled \cite{Friedman1991} function $b(X_i) = \sin\left(\pi X_{i,1}X_{i,2}\right) + 2\left(X_{i,3}-0.5\right)^2 + X_{i,4}+ 0.5X_{i,5}$. For the propensity score we follow \cite{Kunzel2019} and set $p(X_i) = \alpha \left(1+\beta_{2,4}\left(\min(X_{i,1},X_{i,2}) \right)\right)$ where $\beta_{2,4}(\cdot)$ is the beta cumulative distribution function with shape parameters 2 and 4. The share of treated is $\E[D_i] = (31/21) \alpha$ and the true ATE is 1. The innovation standard deviation is set to $\sigma=1$, and results for $\sigma=5$ are relegated to Appendix \ref{app:additional_results}.

In the second strategy, we analyze the CU-DML estimator's performance using an Empirical Monte Carlo Study (EMCS) approach \citep{Huber2013,Lechner2013}. Our EMCS, following \cite{Knaus2022}, uses a dataset from \cite{Lechner2020} of Swiss unemployed individuals in 2003. The dataset includes data on the effect of a job search program (treatment) on the number of months employed in the first six months after the start of the program (outcome), and 49 covariates providing information on individual characteristics of the unemployed, the regional employment agency, and regional labour market characteristics. For more data details, see \cite{Knaus2022}. The EMCS proceeds as follows. We estimate the propensity score using a logistic regression on the entire sample of 91'339 unemployed individuals. Then, we define a new sample consisting of only the non-treated individuals whose fitted propensity score $\hat{p}_i$ lies between $0.05$ and $0.95$. In each simulation, we draw a random sample with replacement of size $N$ from this new sample and randomly assign treatments as $D_i = \mathds{1}{\{V_i < \hat{p}_i/\lambda\}}$, where $\hat{p}_i$ is the fitted propensity score from the logistic regression, $V_i \sim \text{Unif}[0,1]$ and $\lambda$ controls the share of treated $\E[D_i]$. The true ATE is 0 by construction.

We compare the CU-DML estimator with the DML estimator \citep{Chernozhukov2018}, and three popular adjustments of the DML estimator that are designed to improve its finite sample performance when the treatment assignment is unbalanced.\footnote{Over the past two decades, these adjustments have been primarily employed to improve the finite sample performance of the inverse-probability-weighted estimator when the treatment assignment is unbalanced; see, among others, the extensive simulation study of \cite{Huber2013}. More recently, these approaches have also been employed to improve the performance of the DML estimator.} The W-DML estimator uses winsorized propensity scores at 0.01 and 0.99 \citep{Imbens2004}. In the N-DML estimator we normalize the weights $w_i^{\mathcal{I}_{-k}} := D_i/\widehat{p}^{\mathcal{I}_{-k}}(X_i)$ to sum to unity \citep[e.g.][]{Wooldridge2018}. Following \cite{Huber2013}, the T-DML estimator normalizes the weights $w_i^{\mathcal{I}_{-k}}$ and then truncates them to not exceed 0.04, before normalizing them again. Finally, we also consider the U-DML estimator, where the entire sample is undersampled and the DML estimator is applied to the undersampled data. All estimators use 5-fold cross-fitting. The nuisance functions are estimated using random forests \citep{Breiman2001} with 500 trees. The maximal depth of the trees and the minimal number of observations in the leaf nodes are determined by 5-fold cross-validation over 20 simulation replications and set to the most frequently selected values.\footnote{Properely tuning the hyperparameters of the random forests is crucial, especially for DML when the treatment assignment is unbalanced. The popular Python implementation of random forest, \texttt{scikit-learn} \citep{scikit-learn}, sets the default minimal number of observations in the leaf nodes to 1. This may lead to highly unstable propensity scores when the treatment assignment is unbalanced. See \cite{Bach2024} for a recent simulation study on the importance of hyperparameter-tuning for the DML estimator.} Details on the software used for the simulations are provided in Appendix \ref{appsub:implementation}.

\begin{table}[!tb]
    \centering
    \begin{threeparttable}
        \caption{Statistics of the two simulation strategies}
        \label{tab:simulation_results}
        \scriptsize
        \begin{tabular}{rrrrr|rrrr|rrrr}
            \multicolumn{13}{c}{\textbf{Panel A:} Synthetic DGP}\\
            \toprule
            & \multicolumn{4}{c}{$N=2000$, $\E[D_i]=2.5\%$, $\sigma=1$} & \multicolumn{4}{c}{$N=4000$, $\E[D_i]=2.5\%$, $\sigma=1$} & \multicolumn{4}{c}{$N=8000$, $\E[D_i]=2.5\%$, $\sigma=1$} \\
            & RMSE & Bias & Std. dev. & Coverage & RMSE & Bias & Std. dev. & Coverage & RMSE & Bias & Std. dev. & Coverage \\
           \midrule
           DML & - & - & - & - & 0.188 & 0.008 & 0.188 & 0.984 & 0.094 & 0.005 & 0.094 & 0.979 \\
           U-DML & - & - & - & - & 0.173 & 0.070 & 0.158 & 0.928 & 0.122 & 0.057 & 0.108 & 0.924 \\
           CU-DML & - & - & - & - & \textbf{0.129} & 0.038 & 0.123 & 0.938 & \textbf{0.086} & 0.018 & 0.084 & 0.953 \\
           W-DML & - & - & - & - & 0.145 & 0.031 & 0.142 & 0.971 & 0.091 & 0.010 & 0.091 & 0.978 \\
           N-DML & - & - & - & - & 0.149 & 0.028 & 0.146 & 0.937 & 0.089 & 0.012 & 0.088 & 0.943 \\
           T-DML & - & - & - & - & 0.145 & 0.035 & 0.141 & 0.942 & 0.089 & 0.012 & 0.088 & 0.943 \\
           \toprule
           & \multicolumn{4}{c}{$N=2000$, $\E[D_i]=5.0\%$, $\sigma=1$} & \multicolumn{4}{c}{$N=4000$, $\E[D_i]=5.0\%$, $\sigma=1$} & \multicolumn{4}{c}{$N=8000$, $\E[D_i]=5.0\%$, $\sigma=1$} \\
           & RMSE & Bias & Std. dev. & Coverage & RMSE & Bias & Std. dev. & Coverage & RMSE & Bias & Std. dev. & Coverage \\
          \midrule
          DML & 0.160 & 0.022 & 0.158 & 0.985 & 0.089 & 0.009 & 0.089 & 0.976 & 0.060 & 0.002 & 0.060 & 0.962 \\
          U-DML & 0.176 & 0.071 & 0.161 & 0.929 & 0.118 & 0.051 & 0.107 & 0.937 & 0.091 & 0.047 & 0.077 & 0.910 \\
          CU-DML & \textbf{0.132} & 0.039 & 0.126 & 0.946 & \textbf{0.086} & 0.017 & 0.084 & 0.960 & 0.059 & 0.008 & 0.059 & 0.950 \\
          W-DML & 0.156 & 0.023 & 0.155 & 0.985 & 0.089 & 0.009 & 0.089 & 0.976 & 0.060 & 0.002 & 0.060 & 0.962 \\
          N-DML & 0.143 & 0.035 & 0.138 & 0.934 & 0.086 & 0.014 & 0.085 & 0.957 & \textbf{0.059} & 0.004 & 0.059 & 0.940 \\
          T-DML & 0.143 & 0.036 & 0.138 & 0.935 & 0.086 & 0.014 & 0.085 & 0.956 & \textbf{0.059} & 0.004 & 0.059 & 0.940 \\
          \toprule
          & \multicolumn{4}{c}{$N=2000$, $\E[D_i]=10.0\%$, $\sigma=1$} & \multicolumn{4}{c}{$N=4000$, $\E[D_i]=10.0\%$, $\sigma=1$} & \multicolumn{4}{c}{$N=8000$, $\E[D_i]=10.0\%$, $\sigma=1$} \\
          & RMSE & Bias & Std. dev. & Coverage & RMSE & Bias & Std. dev. & Coverage & RMSE & Bias & Std. dev. & Coverage \\
         \midrule
         DML & 0.097 & 0.016 & 0.095 & 0.952 & 0.060 & 0.009 & 0.060 & 0.961 & 0.043 & 0.002 & 0.043 & 0.943 \\
         U-DML & 0.126 & 0.057 & 0.113 & 0.924 & 0.088 & 0.046 & 0.075 & 0.915 & 0.069 & 0.042 & 0.055 & 0.867 \\
         CU-DML & \textbf{0.094} & 0.022 & 0.091 & 0.947 & 0.062 & 0.013 & 0.060 & 0.954 & 0.044 & 0.005 & 0.043 & 0.928 \\
         W-DML & 0.097 & 0.016 & 0.095 & 0.952 & 0.060 & 0.009 & 0.060 & 0.961 & 0.043 & 0.002 & 0.043 & 0.943 \\
         N-DML & 0.095 & 0.020 & 0.093 & 0.924 & \textbf{0.060} & 0.010 & 0.059 & 0.950 & \textbf{0.043} & 0.003 & 0.043 & 0.932 \\
         T-DML & 0.095 & 0.020 & 0.093 & 0.924 & \textbf{0.060} & 0.010 & 0.059 & 0.950 & \textbf{0.043} & 0.003 & 0.043 & 0.932 \\
         \bottomrule\\[0.1cm]

         \multicolumn{13}{c}{\textbf{Panel B:} Empirical Monte Carlo Study}\\
         \toprule
         & \multicolumn{4}{c}{$N=2000$, $\E[D_i]=2.5\%$} & \multicolumn{4}{c}{$N=4000$, $\E[D_i]=2.5\%$} & \multicolumn{4}{c}{$N=8000$, $\E[D_i]=2.5\%$} \\
         & RMSE & Bias & Std. dev. & Coverage & RMSE & Bias & Std. dev. & Coverage & RMSE & Bias & Std. dev. & Coverage \\
        \midrule
        DML & - & - & - & - & 0.314 & 0.051 & 0.309 & 0.984 & 0.194 & 0.041 & 0.189 & 0.975 \\
        U-DML & - & - & - & - & 0.261 & 0.024 & 0.260 & 0.950 & 0.183 & 0.013 & 0.183 & 0.962 \\
        CU-DML & - & - & - & - & \textbf{0.208} & 0.016 & 0.208 & 0.946 & \textbf{0.148} & 0.023 & 0.146 & 0.944 \\
        W-DML & - & - & - & - & 0.233 & 0.044 & 0.228 & 0.979 & 0.162 & 0.051 & 0.154 & 0.961 \\
        N-DML & - & - & - & - & 0.253 & 0.057 & 0.246 & 0.948 & 0.174 & 0.050 & 0.166 & 0.940 \\
        T-DML & - & - & - & - & 0.237 & 0.051 & 0.231 & 0.937 & 0.167 & 0.052 & 0.159 & 0.938 \\
        \toprule
        & \multicolumn{4}{c}{$N=2000$, $\E[D_i]=5.0\%$} & \multicolumn{4}{c}{$N=4000$, $\E[D_i]=5.0\%$} & \multicolumn{4}{c}{$N=8000$, $\E[D_i]=5.0\%$} \\
        & RMSE & Bias & Std. dev. & Coverage & RMSE & Bias & Std. dev. & Coverage & RMSE & Bias & Std. dev. & Coverage \\
       \midrule
       DML & 0.256 & 0.049 & 0.252 & 0.977 & 0.172 & 0.043 & 0.167 & 0.962 & 0.111 & 0.032 & 0.107 & 0.959 \\
       U-DML & 0.262 & 0.012 & 0.261 & 0.964 & 0.180 & 0.012 & 0.179 & 0.953 & 0.129 & 0.022 & 0.127 & 0.948 \\
       CU-DML & \textbf{0.213} & 0.015 & 0.213 & 0.958 & \textbf{0.143} & 0.023 & 0.141 & 0.943 & \textbf{0.101} & 0.023 & 0.099 & 0.941 \\
       W-DML & 0.244 & 0.051 & 0.239 & 0.976 & 0.164 & 0.044 & 0.158 & 0.962 & 0.108 & 0.033 & 0.103 & 0.957 \\
       N-DML & 0.233 & 0.056 & 0.226 & 0.950 & 0.163 & 0.048 & 0.156 & 0.941 & 0.110 & 0.036 & 0.104 & 0.941 \\
       T-DML & 0.226 & 0.056 & 0.219 & 0.946 & 0.161 & 0.048 & 0.154 & 0.941 & 0.110 & 0.036 & 0.104 & 0.941 \\
       \toprule
       & \multicolumn{4}{c}{$N=2000$, $\E[D_i]=10.0\%$} & \multicolumn{4}{c}{$N=4000$, $\E[D_i]=10.0\%$} & \multicolumn{4}{c}{$N=8000$, $\E[D_i]=10.0\%$} \\
       & RMSE & Bias & Std. dev. & Coverage & RMSE & Bias & Std. dev. & Coverage & RMSE & Bias & Std. dev. & Coverage \\
      \midrule
      DML & 0.156 & 0.033 & 0.153 & 0.962 & 0.107 & 0.024 & 0.105 & 0.961 & 0.072 & 0.014 & 0.071 & 0.960 \\
      U-DML & 0.184 & 0.016 & 0.183 & 0.962 & 0.128 & 0.017 & 0.127 & 0.955 & 0.088 & 0.014 & 0.087 & 0.952 \\
      CU-DML & \textbf{0.146} & 0.020 & 0.144 & 0.957 & \textbf{0.099} & 0.016 & 0.098 & 0.950 & \textbf{0.069} & 0.012 & 0.068 & 0.949 \\
      W-DML & 0.156 & 0.033 & 0.152 & 0.962 & 0.106 & 0.023 & 0.104 & 0.961 & 0.072 & 0.014 & 0.071 & 0.960 \\
      N-DML & 0.154 & 0.037 & 0.149 & 0.946 & 0.107 & 0.026 & 0.103 & 0.952 & 0.072 & 0.016 & 0.070 & 0.951 \\
      T-DML & 0.153 & 0.038 & 0.149 & 0.945 & 0.106 & 0.025 & 0.103 & 0.952 & 0.072 & 0.016 & 0.070 & 0.951 \\
            \bottomrule
        \end{tabular}
        \begin{tablenotes}
            \small
            \item \textsc{Note}: The table reports the root mean squared error (RMSE), the absolute value of the average bias, the standard deviation, and the coverage of the 95\% confidence interval for the ATE estimators across 1'000 simulations. Panel A reports the results for the synthetic DGP where data is generated according to Equation \eqref{eq:synthetic_dgp}, while Panel B reports the results for the Empirical Monte Carlo Study. For each of the two simulation strategies, the table reports the results for 8 different simulation settings, where the sample size $N$ and the share of treated $\E[D_i]$ are varied. For the synthetic DGP, the innovation standard deviation $\sigma$ is set to 1, which corresponds to a signal-to-noise ratio of $\Var[y]/\sigma^2 = 1.3$. The best performance in terms of RMSE is highlighted in bold.
        \end{tablenotes}
    \end{threeparttable}
\end{table}

The results are presented in Table \ref{tab:simulation_results}. The table reports the root mean squared error (RMSE), the absolute value of the average bias, the standard deviation, and the coverage of the 95\% confidence interval for each ATE estimator. The simulation is repeated 1'000 times for each of the following tuples of sample size and share of treated: $\left(N, \E[D_i]\right)\allowbreak \in \allowbreak\{(2'000, 0.05), (2'000, 0.1), (4'000, 0.025), (4'000, 0.05),(4'000, 0.1), (8'000, 0.025), (8'000, 0.05), (8'000, 0.1)\}$. The case where $N=2'000$ and $\E[D_i]=0.025$ is excluded since we would only have 50 treated observations. The best performance in terms of RMSE in each simulation is highlighted in bold. For the synthetic DGP (Panel A), the lowest RMSE is achieved by the CU-DML estimator in 5 out of 8 settings. Only when at least 400 observations are treated, the CU-DML estimator is outperformed in terms of RMSE by the normalized and trimmed DML estimators. However, the RMSE differences between N-DML, T-DML, and CU-DML estimators are considerably small in these cases. These results are unaffected by a higher innovation standard deviation $\sigma=5$, see Table \ref{tab:simulation_results_sigma5}. For the EMCS (Panel B), CU-DML achieves the smallest RMSE for all combinations of $\left(N, \E[D_i]\right)$. Also for the EMCS, we observe that the outperformance of the CU-DML estimator decreases for larger samples and/or less imbalanced samples. These patterns are in line with the findings from the machine learning classification literature, showing that class imbalancedness is a relative problem, related to both the degree of imbalancedness and the sample size \citep{Japkowicz2002}.

The bias and standard deviation of the estimators unveil that, while the analyzed estimators do not differ substantially in terms of their bias, the CU-DML estimator has the lowest standard deviation across all simulation strategies and settings. The simulation study also confirms the theoretical result presented in the previous section. First, the coverage of the proposed estimator is in all strategies and settings close to 95\%. Second, its RMSE decreases at rate $\sqrt{N}$.

\section{Conclusion}
This paper addresses a common finite sample problem of the double-machine learning ATE estimator (DML) \citep{Chernozhukov2018}: in settings with unbalanced treatment assignment estimations of the propensity scores become either too close to zero or one. This causes the DML estimator to become unstable, especially in small samples. We propose a simple yet effective adjustment of the DML estimator (CU-DML) where the machine learning models for the propensity scores are estimated over an undersampled dataset. The resulting propensity score predictions are then calibrated to adjust for the undersampling. 

We provide theoretical results for the CU-DML estimator, showing that it has the same asymptotic distribution as the DML estimator. In particular, the estimator attains the parametric rate of convergence $\sqrt{N}$ and its variance achieves the semi-parametric efficiency bound \citep{Hahn1998}. Furthermore, a small simulation study provides evidence for the finite sample performance of the proposed estimator, showing that it is of particular use in settings with highly unbalanced treatment assignments or small samples.

Future research could adapt the proposed approach to other estimators that rely on the estimation of the propensity score, such as the inverse probability-weighted estimator. Furthermore, CU-DML could be extended to estimators of the conditional average treatment effect \citep{Fan2022,Zimmert2019}.

\newpage
\bibliographystyle{ecta}
\bibliography{references}

\newpage
\begin{appendix}
    \section{Appendix: Proofs}\label{app:proofs}
    \renewcommand{\theequation}{\thesection\arabic{equation}}
    \setcounter{equation}{0}
    \renewcommand{\thelemma}{\thesection\arabic{lemma}}

    \begin{lemma}
        \label{lemma:convergence_share}
        Let $\gamma = \E[D]/(1-\E[D])$ and $\widehat{\gamma} = \left(\sum_{i=1}^N D_i\right)/\left(1-\sum_{i=1}^N D_i\right)$. Under Assumption \ref{assumption:propensity_score} it holds that:
        \begin{equation*}
            \vert \widehat{\gamma} - \gamma \vert = o_p(N^{-1/2})
        \end{equation*}
    \end{lemma}

\begin{proof}[Proof of Lemma \ref{lemma:convergence_share}]
    Let $\widehat{p} = N^{-1} \sum_{i=1}^N D_i$ and consequently $\widehat{\gamma} = \widehat{p}/(1-\widehat{p}).$ Since $\E[\widehat{p}] = \E[D] = p$, we have that:
    \begin{align*}
        \E[\vert \widehat{p} - p \vert^2] = \Var[\widehat{p}] = \frac{p(1-p)}{N} = o(N^{-1})
    \end{align*}
    and therefore $\vert \widehat{p} - p \vert = o_p(N^{-1/2})$. For $\widehat{\gamma}$ we have:
    \begin{equation*}
        \vert \widehat{\gamma} - \gamma \vert = \left\vert \frac{\widehat{p}}{1-\widehat{p}} - \frac{p}{1-p} \right\vert = \left\vert \frac{\widehat{p}(1-p) - p(1-\widehat{p})}{(1-\widehat{p})(1-p)} \right\vert \leq  \frac{1}{\epsilon^2} \vert \widehat{p} - p \vert = o_p(N^{-1/2})
    \end{equation*}
\end{proof}

\begin{proof}[Proof of Theorem \ref{theorem:convergence}]
    The proof follows the approach of \cite{Wager2022}. If the true nuisance functions $\mu_d(X)$ and $p(X)$ are known, the oracle estimator for $\theta$ given by:
    \begin{align*}
        \widetilde{\theta} &= \frac{1}{N} \sum_{i=1}^N \tau(Z_i)\\
        &= \frac{1}{N} \sum_{i=1}^N \left(\mu_1(X_i) - \mu_0(X_i) + \frac{D_i}{p(X_i)}(Y_i - \mu_1(X_i)) - + \frac{1-D_i}{1-p(X_i)}(Y_i - \mu_0(X_i)) \right) 
    \end{align*}
    is a sample average of i.i.d. random variables and by the central limit theorem we have that:
    \begin{equation*}
        \sqrt{N}\left(\widetilde{\theta} - \theta \right) \overset{d}{\longrightarrow} \mathcal{N}\left(0, V^*\right).
    \end{equation*}
    It therefore sufficies to show that $\sqrt{N}\left(\widetilde{\theta} - \cudml \right) = o_p(1)$. Notice that $\theta = \theta_1 - \theta_0$ where $\theta_d = \E[Y^d]$. Moreover, $\widetilde{\theta} = \widetilde{\theta}_1 - \widetilde{\theta}_0$ and $\cudml = \cudml_1 - \cudml_0$. Therefore, it sufficies to show that $\sqrt{N}\left(\widetilde{\theta}_1 - \cudml_1 \right) = o_p(1)$, the result for $\widetilde{\theta}_0$ follows analogously. The estimator can be further re-written as:
    \begin{align*}
        \cudml_1 &= \frac{1}{N}\sum_{k=1}^{K} \sum_{i\in\mathcal{I}_k} \widehat\mu^{\mathcal{I}_{-k}}_1(X_i) + \frac{D_i}{\widehat{p}^{\mathcal{I}_{-k}}(X_i)}(Y_i-\widehat\mu^{\mathcal{I}_{-k}}_1(X_i))\\
        &= \sum_{k=1}^{K} \frac{\vert \mathcal{I}_k \vert}{N} \cudml_{1,\mathcal{I}_k}
    \end{align*}
    where:
    \begin{equation*}
        \cudml_{1,\mathcal{I}_k} = \frac{1}{\vert \mathcal{I}_k \vert} \sum_{i\in\mathcal{I}_k} \widehat\mu^{\mathcal{I}_{-k}}_1(X_i) + \frac{D_i}{\widehat{p}^{\mathcal{I}_{-k}}(X_i)}(Y_i-\widehat\mu^{\mathcal{I}_{-k}}_1(X_i)).
    \end{equation*}
    Analogously, we can defined the oracle estimator $\widetilde{\theta}_{1,\mathcal{I}_k}$. It therefore sufficies to show that $\sqrt{N}\left(\widetilde{\theta}_{1,\mathcal{I}_k} - \cudml_{1,\mathcal{I}_k} \right) = o_p(1)$ for all $k=1,\dots,K$. We can decompose the difference between the oracle and the CU-DML estimator as:
    \begin{align*}
        \widetilde{\theta}_{1,\mathcal{I}_k} - \cudml_{1,\mathcal{I}_k} = &\underbrace{\frac{1}{\vert \mathcal{I}_k \vert} \sum_{i\in\mathcal{I}_k}  \left( \widehat\mu^{\mathcal{I}_{-k}}_1(X_i) - \mu_1(X_i) \right)\left(1-\frac{D_i}{p(X_i)} \right)}_{\text{(A)}}\\
        +& \underbrace{\frac{1}{\vert \mathcal{I}_k \vert} \sum_{i\in\mathcal{I}_k} D_i \left(\left(Y_i - \mu_1(X_i)\right) \left( \frac{1}{\widehat{p}^{\mathcal{I}_{-k}}(X_i)} - \frac{1}{p(X_i)}\right) \right)}_{\text{(B)}}\\
        -& \underbrace{\frac{1}{\vert \mathcal{I}_k \vert} \sum_{i\in\mathcal{I}_k} D_i \left(\left(\widehat\mu^{\mathcal{I}_{-k}}_1(X_i)  - \mu_1(X_i)\right) \left( \frac{1}{\widehat{p}^{\mathcal{I}_{-k}}(X_i)} - \frac{1}{p(X_i)}\right) \right)}_{\text{(C)}}.
    \end{align*}
    We will show that each of the three terms converges to zero in probability at rate $N^{-1/2}$. Term (A) is not dependent on the estimated propensity score and is not affected by undersampling. As such, the usual arguments apply for its convergence and we refer to \cite{Wager2022} for the details. For term (B) we can show that its squared $L_2$-norm is $o\left(N^{-1/2}\right)$:
    \begin{align}\label{eq:l2_B}
        \begin{split}
            &\norm{\frac{1}{\vert \mathcal{I}_k \vert} \sum_{i\in\mathcal{I}_k} D_i \left(\left(Y_i - \mu_1(X_i)\right) \left( \frac{1}{\widehat{p}^{\mathcal{I}_{-k}}(X_i)} - \frac{1}{p(X_i)}\right) \right)}_2^2\\ 
            &= \E\left[\left(\frac{1}{\vert \mathcal{I}_k \vert} \sum_{i\in\mathcal{I}_k} D_i \left(\left(Y_i - \mu_1(X_i)\right) \left( \frac{1}{\widehat{p}^{\mathcal{I}_{-k}}(X_i)} - \frac{1}{p(X_i)}\right) \right) \right)^2 \right]\\
            &= \E\left[ \E\left[\left(\frac{1}{\vert \mathcal{I}_k \vert} \sum_{i\in\mathcal{I}_k} D_i \left(\left(Y_i - \mu_1(X_i)\right) \left( \frac{1}{\widehat{p}^{\mathcal{I}_{-k}}(X_i)} - \frac{1}{p(X_i)}\right) \right) \right)^2  \Bigg\vert \{Z_i\}_{i\in \mathcal{I}_{-k}},  \{D_i\}_{i=1}^N \right]\right]\\
            &= \E\left[ \Var\left[\frac{1}{\vert \mathcal{I}_k \vert} \sum_{i\in\mathcal{I}_k} D_i \left(\left(Y_i - \mu_1(X_i)\right) \left( \frac{1}{\widehat{p}^{\mathcal{I}_{-k}}(X_i)} - \frac{1}{p(X_i)}\right) \right)  \Bigg\vert \{Z_i\}_{i\in \mathcal{I}_{-k}},  \{D_i\}_{i=1}^N \right]\right]\\
            &= \frac{1}{\vert \mathcal{I}_k \vert^2} \E\left[  \sum_{i\in\mathcal{I}_k} \Var\left[ D_i \left(\left(Y_i - \mu_1(X_i)\right) \left( \frac{1}{\widehat{p}^{\mathcal{I}_{-k}}(X_i)} - \frac{1}{p(X_i)}\right) \right)  \Bigg\vert \{Z_i\}_{i\in \mathcal{I}_{-k}},  \{D_i\}_{i=1}^N \right]\right]\\
            &= \frac{1}{\vert \mathcal{I}_k \vert} \E\left[ \Var\left[ D_i \left(\left(Y_i - \mu_1(X_i)\right) \left( \frac{1}{\widehat{p}^{\mathcal{I}_{-k}}(X_i)} - \frac{1}{p(X_i)}\right) \right)  \Bigg\vert \{Z_i\}_{i\in \mathcal{I}_{-k}},  \{D_i\}_{i=1}^N \right]\right]\\
            &= \frac{1}{\vert \mathcal{I}_k \vert} \E\left[ D_i \left(Y_i - \mu_1(X_i)\right)^2 \left( \frac{1}{\widehat{p}^{\mathcal{I}_{-k}}(X_i)} - \frac{1}{p(X_i)}\right)^2 \right]\\
            &\leq \frac{1}{\vert \mathcal{I}_k \vert} \zeta \E\left[\left( \frac{1}{\widehat{p}^{\mathcal{I}_{-k}}(X_i)} - \frac{1}{p(X_i)}\right)^2 \right] = \frac{o(1)}{N}.
        \end{split}
    \end{align}
    The fourth equality in Equation \eqref{eq:l2_B} follows from the fact that for any $i,j\in \mathcal{I}_k$ with $i\neq j$ we have:
    \begin{align*}
        \Cov\Bigg[& D_i \left(\left(Y_i - \mu_1(X_i)\right) \left( \frac{1}{\widehat{p}^{\mathcal{I}_{-k}}(X_i)} - \frac{1}{p(X_i)}\right) \right),\\ 
        & D_j \left(\left(Y_j - \mu_1(X_j)\right) \left( \frac{1}{\widehat{p}^{\mathcal{I}_{-k}}(X_j)} - \frac{1}{p(X_j)}\right) \right) \Bigg\vert \{Z_i\}_{i\in \mathcal{I}_{-k}},  \{D_i\}_{i=1}^N \Bigg]\\
        = \E\Bigg[& D_i \left(\left(Y_i - \mu_1(X_i)\right) \left( \frac{1}{\widehat{p}^{\mathcal{I}_{-k}}(X_i)} - \frac{1}{p(X_i)}\right) \right)\\ 
        & D_j \left(\left(Y_j - \mu_1(X_j)\right) \left( \frac{1}{\widehat{p}^{\mathcal{I}_{-k}}(X_j)} - \frac{1}{p(X_j)}\right) \right) \Bigg\vert \{Z_i\}_{i\in \mathcal{I}_{-k}},  \{D_i\}_{i=1}^N \Bigg]\\
        = \E\Bigg[& \left( \frac{1}{\widehat{p}^{\mathcal{I}_{-k}}(X_i)} - \frac{1}{p(X_i)}\right) \left( \frac{1}{\widehat{p}^{\mathcal{I}_{-k}}(X_j)} - \frac{1}{p(X_j)}\right) \\ 
        & \E[D_i(Y_i - \mu_1(X_i)) \vert X_i, D_i ] \E[D_j(Y_j - \mu_1(X_j)) \vert X_j, D_j ] \Bigg\vert \{Z_i\}_{i\in \mathcal{I}_{-k}},  \{D_i\}_{i=1}^N \Bigg] = 0
    \end{align*}
    by the law of iterated expectations and independence of the observations.
    The last equality in Equation \eqref{eq:l2_B} follows from Assumption~\ref{assumption:variance} the fact that, under Assumptions \ref{assumption:propensity_score} and \ref{assumption:convergence_ml}, the calibrated propensity score is sup-norm consistent since:
    \begin{align*}
        \vert \widehat{p}^{\mathcal{I}_{-k}}_S(x) \widehat{\gamma} - p_S(x)\gamma  \vert &= \vert \widehat{p}^{\mathcal{I}_{-k}}_S(x) \widehat{\gamma} - p_S(x)\widehat{\gamma} + p_S(x)\widehat{\gamma} - p_S(x)\gamma \vert\\
        & \leq  \vert \widehat{p}^{\mathcal{I}_{-k}}_S(x) - p_S(x) \vert \cdot \vert \widehat{\gamma} \vert + \vert \widehat{\gamma} -\gamma \vert\cdot \vert p_S(x) \vert \\
        & \leq \vert \widehat{p}^{\mathcal{I}_{-k}}_S(x) - p_S(x) \vert + \vert \widehat{\gamma} -\gamma \vert
    \end{align*}
    and $\vert \widehat{\gamma} -\gamma \vert = o_p(N^{-1/2})$ as shown in Lemma \ref{lemma:convergence_share}. We conclude that (B) is $o_p(N^{-1/2})$.\par

    Lastly, we bound the term (C) by showing that its expectation converges to zero at rate $N^{-1/2}$:
    \begin{align}
        \begin{split}
            &\E\left[D_i \left(\left(\widehat\mu^{\mathcal{I}_{-k}}_1(X_i)  - \mu_1(X_i)\right) \left( \frac{1}{\widehat{p}^{\mathcal{I}_{-k}}(X_i)} - \frac{1}{p(X_i)}\right) \right) \right]\\
            & \leq \E\left[\left(\widehat\mu^{\mathcal{I}_{-k}}_1(X_i)  - \mu_1(X_i)\right) \left( \frac{1}{\widehat{p}^{\mathcal{I}_{-k}}(X_i)} - \frac{1}{p(X_i)}\right) \right]\\
            & \leq \E\left[\left\vert \left(\widehat\mu^{\mathcal{I}_{-k}}_1(X_i)  - \mu_1(X_i)\right) \left( \frac{1}{\widehat{p}^{\mathcal{I}_{-k}}(X_i)} - \frac{1}{p(X_i)}\right) \right\vert\right]\\
            &= \E\left[\left\vert \left(\widehat\mu^{\mathcal{I}_{-k}}_1(X_i)  - \mu_1(X_i)\right) \left(\frac{1}{\gamma p_S(X_i)} - \frac{1}{\widehat{\gamma}\widehat{p}_S^{\mathcal{I}_{-k}}(X_i)} -\frac{1}{\gamma} + \frac{1}{\widehat{\gamma}}\right) \right\vert\right]\\
            &\leq \E\left[\left\vert \widehat\mu^{\mathcal{I}_{-k}}_1(X_i)  - \mu_1(X_i)\right\vert \left\vert \left(\frac{1}{\gamma p_S(X_i)} - \frac{1}{\widehat{\gamma}\widehat{p}_S^{\mathcal{I}_{-k}}(X_i)}\right) + \left(\frac{1}{\widehat{\gamma}}-\frac{1}{\gamma}\right) \right\vert\right]\\
            &\leq \E\left[\left\vert \widehat\mu^{\mathcal{I}_{-k}}_1(X_i)  - \mu_1(X_i)\right\vert \left\vert \frac{1}{\gamma p_S(X_i)} - \frac{1}{\widehat{\gamma}\widehat{p}_S^{\mathcal{I}_{-k}}(X_i)}\right\vert\right] + \E\left[\left\vert \widehat\mu^{\mathcal{I}_{-k}}_1(X_i)  - \mu_1(X_i)\right\vert \left\vert \frac{1}{\widehat{\gamma}}-\frac{1}{\gamma}\right\vert\right]\\
            &\leq c_1 \E\left[\left\vert \widehat\mu^{\mathcal{I}_{-k}}_1(X_i)  - \mu_1(X_i)\right\vert \left\vert \widehat{\gamma}\widehat{p}_S^{\mathcal{I}_{-k}}(X_i)-\gamma p_S(X_i)\right\vert\right] + c_2 \E\left[\left\vert \widehat\mu^{\mathcal{I}_{-k}}_1(X_i)  - \mu_1(X_i)\right\vert \left\vert \gamma-\widehat{\gamma}\right\vert\right]\\
            &=c_1 \E\left[\left\vert \widehat\mu^{\mathcal{I}_{-k}}_1(X_i)  - \mu_1(X_i)\right\vert \left\vert \widehat{p}_S^{\mathcal{I}_{-k}}(X_i)- p_S(X_i)\right\vert \widehat{\gamma} \right] \\
            &\phantom{=} + c_1 \E\left[\left\vert \widehat\mu^{\mathcal{I}_{-k}}_1(X_i)  - \mu_1(X_i)\right\vert \left\vert \widehat{\gamma} - \gamma \right\vert p_S(X_i) \right]\\
            &\phantom{=} + c_2 \E\left[\left\vert \widehat\mu^{\mathcal{I}_{-k}}_1(X_i)  - \mu_1(X_i)\right\vert \left\vert \gamma-\widehat{\gamma}\right\vert\right]\\
            &\leq c_1 \E\left[\left\vert \widehat\mu^{\mathcal{I}_{-k}}_1(X_i)  - \mu_1(X_i)\right\vert \left\vert \widehat{p}_S^{\mathcal{I}_{-k}}(X_i)- p_S(X_i)\right\vert \right] \\
            &\phantom{=} + c_1 \E\left[\left\vert \widehat\mu^{\mathcal{I}_{-k}}_1(X_i)  - \mu_1(X_i)\right\vert \left\vert \widehat{\gamma} - \gamma \right\vert \right]\\
            &\phantom{=} + c_2 \E\left[\left\vert \widehat\mu^{\mathcal{I}_{-k}}_1(X_i)  - \mu_1(X_i)\right\vert \left\vert \gamma-\widehat{\gamma}\right\vert\right]\\
            &\leq c_1 \norm{\widehat\mu^{\mathcal{I}_{-k}}_1(X_i)  - \mu_1(X_i)}_2 \norm{\widehat{p}_S^{\mathcal{I}_{-k}}(X_i)- p_S(X_i)}_2 + (c_1 + c_2) \norm{\widehat\mu^{\mathcal{I}_{-k}}_1(X_i)  - \mu_1(X_i)}_2 \norm{\widehat{\gamma} - \gamma}_2\\
            &= o(N^{-1/2}) + o(N^{-1/2}) = o(N^{-1/2})  
        \end{split}
    \end{align}
    where the first term in the last step is $o(N^{-1/2})$ by Assumption \ref{assumption:convergence_ml}~\ref{assumption:convergence_ml_2} and the second term is $o(N^{-1/2})$ by Lemma \ref{lemma:convergence_share} and Assumption \ref{assumption:convergence_ml}~\ref{assumption:convergence_ml_1}. The positive constants $c_1$ and $c_2$ come from the boundness of $p_S(x)$ and $\gamma$ (and that of their estimators). From the law of large numbers, it follows that (C) is $o_p(N^{-1/2})$. 
\end{proof}

\newpage
\section{Appendix: Details on the simulation study}\label{app:additional_results}
\renewcommand{\theequation}{\thesection\arabic{equation}}
\setcounter{equation}{0}
\renewcommand{\thelemma}{\thesection\arabic{lemma}}
\setcounter{lemma}{0}
\renewcommand{\thetable}{\thesection\arabic{table}}
\setcounter{table}{0}

\subsection{Implementation details}\label{appsub:implementation}
The simulation study is implemented in \texttt{Python 3.12}. The data generating process is implemented using the \texttt{numpy 1.26.4} package. The machine learning models are estimated using the \texttt{scikit-learn 1.4.0} package, using the acceleration extension \texttt{scikit-learn-intelex 2024.4.0}. The empirical data is processed using \texttt{pandas 2.2.2}. The code for the simulation study is availabile at \url{https://github.com/dballinari/Calibrating-doubly-robust-estimators-with-unbalanced-treatment-assignment}.

\subsection{Additional simulation results}

\begin{table}[!h]
    \centering
    \begin{threeparttable}
        \caption{Statistics of the synthetic DGP with $\sigma=5$}
        \label{tab:simulation_results_sigma5}
        \scriptsize
        \begin{tabular}{rrrrr|rrrr|rrrr}
            \toprule
            & \multicolumn{4}{c}{$N=2000$, $\E[D_i]=2.5\%$, $\sigma=5$} & \multicolumn{4}{c}{$N=4000$, $\E[D_i]=2.5\%$, $\sigma=5$} & \multicolumn{4}{c}{$N=8000$, $\E[D_i]=2.5\%$, $\sigma=5$} \\
            & RMSE & Bias & Std. dev. & Coverage & RMSE & Bias & Std. dev. & Coverage & RMSE & Bias & Std. dev. & Coverage \\
           \midrule
           DML & - & - & - & - & 0.885 & 0.025 & 0.884 & 0.983 & 0.448 & -0.001 & 0.448 & 0.980 \\
           U-DML & - & - & - & - & 0.761 & 0.097 & 0.755 & 0.948 & 0.523 & 0.065 & 0.519 & 0.954 \\
           CU-DML & - & - & - & - & \textbf{0.588} & 0.056 & 0.585 & 0.944 & \textbf{0.402} & 0.019 & 0.401 & 0.963 \\
           W-DML & - & - & - & - & 0.675 & 0.053 & 0.673 & 0.974 & 0.429 & 0.005 & 0.428 & 0.979 \\
           N-DML & - & - & - & - & 0.698 & 0.046 & 0.696 & 0.932 & 0.420 & 0.008 & 0.420 & 0.948 \\
           T-DML & - & - & - & - & 0.671 & 0.055 & 0.669 & 0.934 & 0.420 & 0.008 & 0.420 & 0.948 \\
           \toprule
           & \multicolumn{4}{c}{$N=2000$, $\E[D_i]=5.0\%$, $\sigma=5$} & \multicolumn{4}{c}{$N=4000$, $\E[D_i]=5.0\%$, $\sigma=5$} & \multicolumn{4}{c}{$N=8000$, $\E[D_i]=5.0\%$, $\sigma=5$} \\
           & RMSE & Bias & Std. dev. & Coverage & RMSE & Bias & Std. dev. & Coverage & RMSE & Bias & Std. dev. & Coverage \\
          \midrule
          DML & 0.742 & 0.038 & 0.741 & 0.982 & 0.421 & 0.020 & 0.420 & 0.973 & 0.287 & 0.008 & 0.287 & 0.961 \\
          U-DML & 0.773 & 0.100 & 0.767 & 0.949 & 0.510 & 0.045 & 0.508 & 0.961 & 0.375 & 0.055 & 0.372 & 0.951 \\
          CU-DML & \textbf{0.600} & 0.056 & 0.597 & 0.950 & \textbf{0.395} & 0.025 & 0.395 & 0.958 & 0.284 & 0.017 & 0.283 & 0.951 \\
          W-DML & 0.727 & 0.034 & 0.726 & 0.982 & 0.420 & 0.021 & 0.420 & 0.973 & 0.287 & 0.008 & 0.287 & 0.961 \\
          N-DML & 0.656 & 0.054 & 0.654 & 0.934 & 0.402 & 0.028 & 0.401 & 0.950 & \textbf{0.283} & 0.012 & 0.283 & 0.947 \\
          T-DML & 0.651 & 0.045 & 0.649 & 0.936 & 0.402 & 0.028 & 0.401 & 0.950 & \textbf{0.283} & 0.012 & 0.283 & 0.947 \\
          \toprule
          & \multicolumn{4}{c}{$N=2000$, $\E[D_i]=10.0\%$, $\sigma=5$} & \multicolumn{4}{c}{$N=4000$, $\E[D_i]=10.0\%$, $\sigma=5$} & \multicolumn{4}{c}{$N=8000$, $\E[D_i]=10.0\%$, $\sigma=5$} \\
          & RMSE & Bias & Std. dev. & Coverage & RMSE & Bias & Std. dev. & Coverage & RMSE & Bias & Std. dev. & Coverage \\
         \midrule
         DML & 0.456 & 0.039 & 0.454 & 0.953 & 0.288 & 0.028 & 0.287 & 0.965 & 0.208 & 0.010 & 0.208 & 0.949 \\
         U-DML & 0.550 & 0.080 & 0.544 & 0.944 & 0.362 & 0.047 & 0.359 & 0.951 & 0.272 & 0.047 & 0.268 & 0.932 \\
         CU-DML & \textbf{0.438} & 0.045 & 0.436 & 0.942 & 0.292 & 0.029 & 0.290 & 0.958 & 0.210 & 0.015 & 0.209 & 0.936 \\
         W-DML & 0.456 & 0.039 & 0.454 & 0.953 & 0.288 & 0.028 & 0.287 & 0.965 & 0.208 & 0.010 & 0.208 & 0.949 \\
         N-DML & 0.446 & 0.043 & 0.444 & 0.924 & \textbf{0.286} & 0.031 & 0.284 & 0.949 & \textbf{0.207} & 0.012 & 0.207 & 0.939 \\
         T-DML & 0.446 & 0.043 & 0.444 & 0.924 & \textbf{0.286} & 0.031 & 0.284 & 0.949 & \textbf{0.207} & 0.012 & 0.207 & 0.939 \\
      \bottomrule
        \end{tabular}
        \begin{tablenotes}
            \small
            \item \textsc{Note}: The table reports the root mean squared error (RMSE), the absolute value of the average bias, the standard deviation, and the coverage of the 95\% confidence interval for the ATE estimators across 1'000 simulations. The data is generated from the synthetic DGP according to Equation \eqref{eq:synthetic_dgp}, with innovation standard deviation $\sigma=5$, which corresponds to a signal-to-noise ratio of $\Var[y]/\sigma^2 = 1.0$. The table reports the results for 8 different simulation settings, where the sample size $N$ and the share of treated $\E[D_i]$ are varied. The best performance in terms of RMSE is highlighted in bold.
        \end{tablenotes}
    \end{threeparttable}
\end{table}

\end{appendix}

\end{document}